\RequirePackage{fix-cm}
\documentclass[smallextended]{svjour3}       
\smartqed  
\usepackage{graphicx}
\usepackage{amssymb}
\usepackage{mathrsfs}
\usepackage{amssymb}

\newcommand{\Tr}{\textrm{\rm Tr}}
\newcommand{\Ord}{\textrm{\rm Ord}}
\newcommand{\CLI}{\mathscr{C}}

\begin{document}

\def\bbbf{{\rm I\!F}}

\title{The $b$-symbol weight distributions of all semiprimitive irreducible cyclic codes}

\titlerunning{The $b$-symbol weight distributions of all semiprimitive irreducible cyclic codes}

\author{Gerardo Vega}

\authorrunning{Gerardo Vega} 

\institute{Gerardo Vega \at
              Direcci\'on General de C\'omputo y de Tecnolog\'{\i}as de Informaci\'on y Comunicaci\'on, Universidad Nacional Aut\'onoma de M\'exico, 04510 Ciudad de M\'exico, MEXICO \\
              \email{gerardov@unam.mx} 
}

\date{Received: date / Accepted: date}

\maketitle

\begin{abstract}
Up to a new invariant $\mu(b)$, the complete $b$-symbol weight distribution of a particular kind of two-weight irreducible cyclic codes, was recently obtained by Zhu et al. [Des. Codes Cryptogr., 90 (2022) 1113-1125]. The purpose of this paper is to simplify and generalize the results of Zhu et al., and obtain the $b$-symbol weight distributions of all one-weight and two-weight semiprimitive irreducible cyclic codes.

\end{abstract}

\section{Introduction}\label{intro}
Up to a new invariant $\mu(b)$, the complete $b$-symbol weight distribution of some irreducible cyclic codes was recently obtained in \cite{Zhu}. The irreducible cyclic codes considered therein, belong to a particular kind of two-weight irreducible cyclic codes. Thus, the purpose of this paper is to present a generalization for the invariant $\mu(b)$, which will allow us to obtain the $b$-symbol weight distributions of all one-weight and two-weight irreducible cyclic codes, excluding only the exceptional two-weight irreducible cyclic codes studied in \cite{Schmidt}.

This work is organized as follows: In Section \ref{secdos}, we fix some notation and recall some definitions and some known results to be used in subsequent sections. Section \ref{sectres} is devoted to presenting preliminary results. Particularly, in this section, we give an alternative proof of an already known result which determines the weight distributions of all one-weight and two-weight semiprimitive irreducible cyclic codes. In Section \ref{seccuatro}, we use such alternative proof, in order to determine the $b$-symbol weight distributions of all one-weight and two-weight semiprimitive irreducible cyclic codes. 
\section{Notation, definitions and known results}\label{secdos}
Unless otherwise specified, throughout this work we will use the following:

\medskip       
\noindent
{\bf Notation.} For integers $v$ and $w$, with $\gcd(v,w)=1$, $\mbox{\Ord}_v(w)$ will denote the {\em multiplicative order} of $w$ modulo $v$. By using $p$, $t$, $q$, $r$, and $\Delta$, we will denote positive integers such that $p$ is a prime number, $q=p^t$ and $\Delta=\frac{q^r-1}{q-1}$. From now on, $\gamma$ will denote a fixed primitive element of $\bbbf_{q^r}$. Let $u$ be an integer such that $u|(q^r-1)$. For $i=0,1,\cdots,u-1$, we define ${\cal C}_i^{(u,q^r)}:=\gamma^i \langle \gamma^u \rangle$, where $\langle \gamma^u \rangle$ denotes the subgroup of $\bbbf_{q^r}^*$ generated by $\gamma^u$. The cosets ${\cal C}_i^{(u,q^r)}$ are called the {\em cyclotomic classes} of order $u$ in $\bbbf_{q^r}$. For an integer $u$, such that $\gcd(p,u)=1$, $p$ is said to be {\em semiprimitive modulo} $u$ if there exists a positive integer $d$ such that $u|(p^d+1)$. Additionally, we will denote by ``$\mbox{\Tr}_{\bbbf_{q^r}/\bbbf_q}$" the {\em trace mapping} from $\bbbf_{q^r}$ to $\bbbf_q$.

\medskip       
\medskip       
\noindent
{\bf Main assumption.} From now on, we are going to use $n$ and $N$ as integers in such a way that $nN=q^r-1$, with the important assumption that $r=\mbox{\Ord}_n(q)$. Under these circumstances, observe that if $h_N(x) \in  \bbbf_{q}[x]$ is the minimal polynomial of $\gamma^{-N}$, then $h_N(x)$ is parity-check polynomial of an irreducible cyclic code of length $n$ and dimension $r$ over $\bbbf_q$.

\medskip       
Denote by $w_H(\cdot)$ the usual Hamming weight function. For $1\leq b<n$, let the Boolean function $\bar{\cal Z}:\bbbf_{q}^b \to \{0,1\}$ be defined by $\bar{\cal Z}(v)=0$ iff $v$ is the zero vector in $\bbbf_{q}^b$. The $b$-{\em symbol Hamming weight}, $w_b(\mathbf{x})$, of $\mathbf{x}=(x_0,\cdots,x_{n-1}) \in \bbbf_{q}^n$ is defined as

$$w_b(\mathbf{x})\!:=\!w_H(\bar{\cal Z}(x_0,\cdots,x_{b-1}),\bar{\cal Z}(x_1,\cdots,x_{b}),\cdots,\bar{\cal Z}(x_{n-1},\cdots,x_{b+n-2 \!\!\!\!\!\pmod{n}}))\:.$$

\noindent
When $b=1$, $w_1(\mathbf{x})$ is exactly the Hamming weight of $\mathbf{x}$, that is $w_1(\mathbf{x})=w_H(\mathbf{x})$. For any $\mathbf{x}, \mathbf{y} \in \bbbf_{q}^n$, we define the b-{\em symbol distance} (b-distance for short) between $\mathbf{x}$ and $\mathbf{y}$, $d_b(\mathbf{x},\mathbf{y})$, as $d_b(\mathbf{x},\mathbf{y}):=w_b(\mathbf{x}-\mathbf{y})$, and for a code $\CLI$, the b-{\em symbol minimum Hamming distance}, $d_b(\CLI)$, of $\CLI$ is defined as $d_b(\CLI):=\mbox{min } d_b(\mathbf{x},\mathbf{y})$, with $\mathbf{x},\mathbf{y} \in \CLI$ and $\mathbf{x} \neq \mathbf{y}$. Let $A_i^{(b)}$ denote the number of codewords with $b$-symbol Hamming weight $i$ in a code $\CLI$ of length $n$. The $b$-{\em symbol Hamming weight enumerator} of $\CLI$ is defined by

$$1+A_1^{(b)}T+A_2^{(b)}T^2+\cdots+A_n^{(b)}T^n\;.$$

\noindent
Note that if $b=1$, then the $b$-symbol Hamming weight enumerator of $\CLI$ is the ordinary Hamming weight enumerator of $\CLI$. If $\CLI$ is an $(n,M,d_b(\CLI))_q$ $b$-symbol code, with $b\leq d_b(\CLI)\leq n$, then Ding et al. \cite{Ding2} established the Singleton-type bound $M\leq q^{n-d_b(\CLI)+b}$. An $(n,M,d_b(\CLI))_q$ $b$-symbol code $\CLI$ with $M=q^{n-d_b(\CLI)+b}$ is called a {\em maximum distance separable} (MDS for short) $b$-symbol code.

The following gives an explicit description of an irreducible cyclic code of length $n$ and dimension $r$ over $\bbbf_q$ (recall that $nN=q^r-1$ and $r=\mbox{\Ord}_n(q)$).

\begin{definition}\label{defcero} 
Let $q$, $r$, $n$, and $N$ be as before. Then the set 

\[\CLI:=\{(\mbox{\Tr}_{\bbbf_{q^r}/\bbbf_q}(a \gamma^{Ni}))_{i=0}^{n-1} \: | \: a \in \bbbf_{q^r} \} \; ,\]

\noindent
is called an {\em irreducible cyclic code} of length $n$ and dimension $r$ over $\bbbf_q$.
\end{definition} 

An important kind of irreducible cyclic codes are the so-called {\em semiprimitive irreducible cyclic codes}:

\begin{definition}\label{defsemi}
\cite[Definition 4]{Vega} With our current notation and main assumption, fix $u=\gcd(\Delta,N)$. Then, any $[n,r]$ irreducible cyclic code over $\bbbf_{q}$ is semiprimitive  if $u \geq 2$ and the prime $p$ is semiprimitive modulo $u$.
\end{definition}

Apart from a few exceptional codes, it is well known that all two-weight irreducible cyclic codes are semiprimitive. In fact, it is conjectured in \cite{Schmidt} that the number of these exceptional codes is eleven.

The {\em canonical additive character} of $\bbbf_q$ is defined as follows:

$$\chi(x):=e^{2\pi \sqrt{-1}\Tr(x)/p}\ , \ \ \ \  \mbox{ for all } x \in \bbbf_{q} \ ,$$

\noindent
where ``Tr" denotes the trace mapping from $\bbbf_{q}$ to the prime field $\bbbf_p$. Let $a\in\bbbf_q$. The orthogonality relation for the canonical additive character $\chi$ of $\bbbf_q$ is given by (see for example \cite[Chapter 5]{Lidl}):

$$\sum_{x \in \bbbf_{q}} \chi(ax)=\left\{ \begin{array}{cl}
		\;q\; & \mbox{ if $a=0$,} \\
\\
		\;0\; & \mbox{ otherwise.}
			\end{array}
\right .$$

\noindent
This property plays an important role in numerous applications of finite fields. Among them, this property is useful for determining the Hamming weight of a given vector over a finite field; for example if $V=(a_0,a_1,\cdots,a_{n-1}) \in \bbbf_{q}^n$, then

\begin{eqnarray}\label{eqOrtbis}
w_H(V)=n-\frac{1}{q}\sum_{i=0}^{n-1}\sum_{y \in \bbbf_{q}}\chi(ya_i)\;.
\end{eqnarray}

Let $\chi'$ be the canonical additive character of $\bbbf_{q^r}$ and let $u \geq 1$ be an integer such that $u|(q^r-1)$. For $i=0,1,\cdots,u-1$, the $i$-th {\em Gaussian period}, $\eta_i^{(u,q^r)}$, of order $u$ for $\bbbf_{q^r}$ is defined to be

$$\eta_i^{(u,q^r)}:=\sum_{x \in {\cal C}_i^{(u,q^r)}} \chi'(x) \; .$$

\noindent
Suppose that $a\in {\cal C}_i^{(u,q^r)}$. Since $\sum_{x \in \bbbf_{q^r}}\chi'(ax^u)=u\eta_i^{(u,q^r)}+1$ and $\eta_0^{(1,q^r)}+1=0$, the following result is a direct consequence of Theorem 1 in \cite{Moisio}: 

\begin{theorem}\label{teouno}
With our notation suppose that $rt=2sd$ and $u | (p^d+1)$, for positive integers $s$, $d$ and $u$. Then 

$$\frac{u\eta_i^{(u,q^r)} + 1}{q^{r/2}}= \left\{ \begin{array}{cl}
		(-1)^{s-1}(u-1) & \mbox{ if } i \equiv \delta \pmod{u} \:, \\
\\
		(-1)^{s} & \mbox{ if } i \not\equiv \delta \pmod{u} \:,
			\end{array}
\right .$$

\noindent
where the integer $\delta$ is defined in terms of the following two cases:

$$\delta:=
\left\{ \begin{array}{cl}
		0\; & \mbox{if $u=1$; or $p=2$; or $p>2$ and $2|s$; or $p>2$, $2 \nmid s$, and $2|\frac{p^d+1}{u}$} \:, \\
		  &  \\
		\frac{u}{2}\; & \mbox{if $p>2$, $2 \nmid s$ and $2 \nmid \frac{p^d+1}{u}$} \:.
			\end{array}
\right .$$
\end{theorem} 

\begin{remark}\label{rmMoisio} 
As shown below, by means of the previous theorem, it is possible to determine, in a single result, the Hamming weight enumerator of all one-weight and semiprimitive two-weight irreducible cyclic codes. 
\end{remark}

In certain circumstances it is necessary to consider the set of products of the form $xy$, where $x \in {\cal C}_i^{(N,q^r)}$, $y \in \bbbf_{q}^*$, and $0 \leq i < N$. The following result goes in this direction:

\begin{lemma}\label{defmultiset}
\cite[Lemma 5]{Ding1} Let $N$ be a positive divisor of $q^r-1$ and let $i$ be any integer with $0 \leq i < N$. Fix $u=\gcd(\Delta,N)$. We have the following multiset equality:

$$\left\{ xy : x \in {\cal C}_i^{(N,q^r)}, \;y \in \bbbf_{q}^* \right\}=\frac{(q-1)u}{N} * {\cal C}_i^{(u,q^r)} \; ,$$

\noindent
where $\frac{(q-1)u}{N} * {\cal C}_i^{(u,q^r)}$ denotes the multiset in which each element in the set ${\cal C}_i^{(u,q^r)}$ appears in the multiset with multiplicity $\frac{(q-1)u}{N}$.
\end{lemma} 

The following definitions are inspired by and similar to those of \cite{Zhu}.

\begin{definition}\label{defuno} Let $b$ be an integer, with $1\leq b\leq r$. Let ${\cal P}(b)$ be the subset of cardinality $(q^b-1)/(q-1)$ in $\bbbf_{q^r}^*$ defined as

$${\cal P}(b):=\bigcup_{j=1}^{b-1} \{\gamma^{(j-1)N}+x_1\gamma^{jN}+\cdots+x_{b-j}\gamma^{(b-1)N} : x_1,\cdots,x_j \in \bbbf_q \} \cup \{\gamma^{(b-1)N}\} \;.$$
\end{definition} 

\begin{remark}\label{rmpb} 
Note that ${\cal P}(1)=\{1\}$.
\end{remark}

\begin{definition}\label{defdos} 
Let $b$ be as in Definition \ref{defuno} and fix $u=\gcd(\Delta,N)$. For $0\leq i < u$, we define $\mu_{(i)}(b)$ as

$$\mu_{(i)}(b):=|\{ x \in {\cal P}(b) : x \in {\cal C}_i^{(u,q^r)} \}| \;.$$
\end{definition} 

\begin{remark}\label{rminvariant} 
Since ${\cal C}_0^{(2,q^r)}=\{ x\in \bbbf_{q^r}^* : x \mbox{ is a square in } \bbbf_{q^r}^*\}$, note that $\mu_{(i)}(b)$ is indeed a generalization of the invariant $\mu(b)$ in \cite{Zhu}. Furthermore, note that $\mu_{(0)}(1)=1$ and $\mu_{(i)}(1)=0$, for $1 \leq i < u$.  
\end{remark}

The following important result from \cite{Zhu}, is key in order to achieve our goals.

\begin{lemma}\label{lemauno}
\cite[Lemma 4.3]{Zhu}
Let $\CLI$ be as in Definition \ref{defcero} and let $c(a) \in \CLI$ be a codeword. Then, for any integer $1\leq b\leq r$,

$$w_b(c(a))=\frac{1}{q^{b-1}}\sum_{\theta \in {\cal P}(b)}w_H(c(\theta a)) \;.$$
\end{lemma} 

\begin{remark}\label{rmcero} 
The previous lemma is key for us, because although the condition $\gcd(\frac{q^r-1}{q-1},N)=2$ is one of the main assumptions in \cite{Zhu}, Lemma 4.3 is beyond such condition. However it is important to observe that there is a small misprint in the proof of Lemma 4.3; more specifically the equality

$$n-w_1(c(a))=\sum_{x \in I} \frac{1}{q} \sum_{y\in\bbbf_{q}} \chi(yax)\;,$$

\noindent
should be

$$n-w_1(c(a))=\sum_{x \in I} \frac{1}{q} \sum_{y\in\bbbf_{q}} \chi(yax^N)\;.$$  
\end{remark}

\section{Preliminary results}\label{sectres}

In the light of Remark \ref{rminvariant}, the following is a generalization of \cite[Lemma 2.1]{Zhu}.

\begin{lemma}\label{lemados} 
Let $b$ and $\mu_{(i)}(b)$ be as in Definition \ref{defdos}. If $b=r$ then, for any $0\leq i < u$, we have

$$\mu_{(i)}(r)=\frac{1}{u}|{\cal P}(b)|=\frac{\Delta}{u}\;.$$
\end{lemma} 

\begin{proof}
Clearly

$$\bbbf_{q^r}^*=\bigsqcup_{x \in {\cal P}(b)} x\bbbf_{q}^* \;,$$

\noindent
where $\sqcup$ is a disjoint union. Now, since $u| \Delta$ and $\langle \gamma^{\Delta} \rangle=\bbbf_{q}^*$, $x \in {\cal C}_i^{(u,q^r)}$ if and only if each element of $x\bbbf_{q}^*$ is also in ${\cal C}_i^{(u,q^r)}$. This implies that

$$\mu_{(i)}(r)(q-1)=\frac{q^r-1}{u}\;,$$

\noindent
which is the number of elements in ${\cal C}_i^{(u,q^r)}$. This completes the proof. \qed
\end{proof} 

It is already known the Hamming weight enumerator of all one-weight and semiprimitive two-weight irreducible cyclic codes over any finite field (see for example \cite{Vega}). By means of the following theorem we recall such a result and give an alternative proof of it. As will be pointed out below, this alternative proof will be important for fulfilling our goals.

\begin{theorem}\label{teomio1}
Let $\CLI$ be as in Definition \ref{defcero}. Fix $u=\gcd(\Delta,N)$. Assume that $u=1$ or $p$ is semiprimitive modulo $u$. Let $d$ be the smallest positive integer such that $u | (p^d+1)$ and let $s=1$ if $u=1$ and $s=(rt)/(2d)$ if $u>1$. Fix

$$W_A=\frac{n q^{r/2-1}}{\Delta}(q^{r/2}-(-1)^{s-1}(u-1))\;\;\; \mbox{ and } \;\;\;W_B=\frac{n q^{r/2-1}}{\Delta}(q^{r/2}-(-1)^s)\;.$$ 

\noindent
Then, $\CLI$ is an $[n,r]$ irreducible cyclic code whose Hamming weight enumerator is 

\begin{equation}\label{eqtres}
1+\frac{q^r-1}{u}T^{W_A}+\frac{(q^r-1)(u-1)}{u}T^{W_B} \; .
\end{equation}
\end{theorem}

\begin{remark}\label{rmone-weight} 
Note that Theorem \ref{teomio1} gives, in a single result, an explicit description of the Hamming weight enumerators of all one-weight ($u=1$) and two-weight ($2\leq u < \Delta$) irreducible cyclic codes, excluding only the exceptional two-weight irreducible cyclic codes studied in \cite{Schmidt}. Therefore observe that the two-weight irreducible cyclic codes considered in \cite{Zhu} ($u=\gcd(\Delta,N)=2$) belong also to Theorem \ref{teomio1}.
\end{remark}

\begin{proof}
First note that if $u>1$, then there must exist an integer $s$ such that $rt=2sd$. 

For $a \in \bbbf_{q^r}^*$, let $c(a) = (\mbox{\Tr}_{\bbbf_{q^r}/\bbbf_q}(a \gamma^{Ni}))_{i=0}^{n-1} \in \CLI$. Let $\chi$ and $\chi'$ be the canonical additive characters of $\bbbf_{q}$ and $\bbbf_{q^r}$, respectively. Thus, by the orthogonality relation for the character $\chi$ (see (\ref{eqOrtbis})) the Hamming weight of the codeword $c(a)$, $w_H(c(a))$, is 

\begin{eqnarray*}
w_H(c(a)) &=& n-\frac{1}{q}\sum_{i=0}^{n-1}\sum_{y\in\bbbf_{q}}\chi(y\mbox{\Tr}_{\bbbf_{q^r}/\bbbf_q}(a \gamma^{Ni})) \ ,\\
&=&n-\frac{n}{q}-\frac{1}{q}\sum_{y\in\bbbf_{q}^*}\sum_{x\in {\cal C}_0^{(N,q^r)}}\chi'(yax) \ ,\\
&=&n-\frac{n}{q}-\frac{(q-1)u}{qN}\sum_{z\in {\cal C}_0^{(u,q^r)}}\chi'(az) \ ,
\end{eqnarray*}

\noindent
where the last equality holds by Lemma \ref{defmultiset}. Now, suppose that $a \in {\cal C}_i^{(u,q^r)}$ for some $0\leq i < u$. Thus

\begin{eqnarray*}
w_H(c(a)) &=& n-\frac{n}{q}-\frac{(q-1)}{qN}u \eta_i^{(u,q^r)} \ ,\\
&=&\frac{n}{\Delta q}(q^r-1)-\frac{n}{\Delta q}u \eta_i^{(u,q^r)} \ ,\\
&=&\frac{nq^{r-1}}{\Delta}-\frac{n}{\Delta q}(u \eta_i^{(u,q^r)}+1) \ ,\\
&=&\frac{nq^{r-1}}{\Delta}-\frac{n q^{r/2-1}}{\Delta}\frac{(u\eta_i^{(u,q^r)} + 1)}{q^{r/2}} \ ,\\
&=&\frac{n q^{r/2-1}}{\Delta}(q^{r/2}-\frac{u\eta_i^{(u,q^r)} + 1}{q^{r/2}}) \ .
\end{eqnarray*}

\noindent
Let $\delta$ be as in Theorem \ref{teouno} and observe that $i \equiv \delta \pmod{u}$ iff $a \in {\cal C}_{\delta}^{(u,q^r)}$. Therefore, owing to Theorem \ref{teouno}, we have

\begin{equation}\label{eqwb0}
w_H(c(a))=
\left\{ \begin{array}{cl}
		W_A\; & \mbox{if $a \in {\cal C}_{\delta}^{(u,q^r)}$}\:, \\
		  &  \\
		W_B\; & \mbox{if $a \in \bbbf_{q^r}^* \setminus {\cal C}_{\delta}^{(u,q^r)}$ }\:.
			\end{array}
\right .
\end{equation}

\noindent
The result now follows from the fact that $|{\cal C}_{\delta}^{(u,q^r)}|=\frac{q^r-1}{u}$ and $|\bbbf_{q^r}^* \setminus{\cal C}_{\delta}^{(u,q^r)}|=\frac{(q^r-1)(u-1)}{u}$. \qed
\end{proof}

\section{The $b$-symbol weight distribution of all one-weight and two-weight semiprimitive irreducible cyclic codes}\label{seccuatro}

We are now in conditions to present our main results.

\begin{theorem}\label{teomio2}
Assume the same notation and assumptions as in Theorem \ref{teomio1}. Let ${\cal P}(b)$, $\mu_{(i)}(b)$, and $\delta$ be as before. For $0 \leq i < u$ and $1 \leq b \leq r$, let 

\begin{equation}\label{eqwb1}
W_i^{(b)} = \frac{(q-1)q^{r/2-b}}{N}\left[|{\cal P}(b)|\left(q^{r/2}-(-1)^s\right)+(-1)^s u\mu_{((\delta-i)\!\!\!\!\!\pmod{u})}(b)\right] \ .
\end{equation}

\noindent
Then, the $b$-symbol Hamming weight enumerator of $\CLI$ is

\begin{equation}\label{eqwb2}
A(T)=1+\frac{q^r-1}{u}\sum_{i=0}^{u-1}T^{W_i^{(b)}}\;.
\end{equation}
\end{theorem}

\begin{proof}
Let $a\in \bbbf_{q^r}^*$ and let $c(a) \in \CLI$. Let $W_A$ and $W_B$ be as in Theorem \ref{teomio1} and suppose that $a\in {\cal C}_i^{(u,q^r)}$, for some $0 \leq i < u$. Thus, from (\ref{eqwb0}), $w_H(c(\theta a))=W_A$ iff $\theta a \in {\cal C}_\delta^{(u,q^r)}$ iff $\theta \in {\cal C}_{(\delta-i)\!\!\!\!\!\pmod{u}}^{(u,q^r)}$. But there are exactly $\mu_{((\delta-i)\!\!\!\!\!\pmod{u})}(b)$ elements $\theta$ in ${\cal P}(b)$ that satisfy the condition $\theta \in {\cal C}_{(\delta-i)\!\!\!\!\!\pmod{u}}^{(u,q^r)}$. Therefore, owing to Lemma \ref{lemauno}, $w_b(c(a))=W_i^{(b)}$ where

$$W_i^{(b)}=\frac{1}{q^{b-1}}\left[\mu_{((\delta-i)\!\!\!\!\!\pmod{u})}(b)W_A+\left(|{\cal P}(b)|-\mu_{((\delta-i)\!\!\!\!\!\pmod{u})}(b)\right)W_B \right]\;.$$

\noindent
Hence, (\ref{eqwb1}) follows by considering the explicit values of $W_A$ and $W_B$ in Theorem \ref{teomio1}. Finally, the $b$-symbol Hamming weight enumerator of $\CLI$ follows from the fact that $|{\cal C}_i^{(u,q^r)}|=\frac{q^r-1}{u}$, for any $0\leq i < u$. \qed
\end{proof} 

Note that the previous theorem is also valid for $b=1$. In fact, in this case, the ordinary Hamming weight enumerator in (\ref{eqtres}) is exactly the same as the $1$-symbol Hamming weight enumerator of (\ref{eqwb2}) (take into consideration Remarks \ref{rmpb} and \ref{rminvariant}). Therefore we see that Theorem \ref{teomio2} not only simplifies and generalizes \cite[Corollary 3.1]{Zhu} but also generalizes Theorem \ref{teomio1}.

\begin{example}\label{ejemuno}
The following are some examples of Theorem \ref{teomio2}.

\begin{enumerate}
\item[{\rm (a)}] Let $(q,r,N,b)=(3,4,2,3)$. Thus $u=\gcd(\Delta,N)=2$, $s=2$, $\delta=0$, and $|{\cal P}(b)|=q^2+q+1=13$. Since $\mu_{(0)}(b)=8$ (see \cite[Example 2.3]{Zhu}), $\mu_{(1)}(b)=|{\cal P}(b)|-\mu_{(0)}(b)=5$. Therefore, owing to Theorems \ref{teomio1} and \ref{teomio2}, $W_A=30$, $W_B=24$, $W_0^{(3)}=40$, $W_1^{(3)}=38$, and $\CLI$ is a $[40,4]_3$ irreducible cyclic code whose ordinary and $3$-symbol Hamming weight enumerators are $1+40T^{24}+40T^{30}$ and $1+40T^{38}+40T^{40}$, respectively.

\item[{\rm (b)}] Let $(q,r,N,b)=(2,4,3,2)$. Thus $u=\gcd(\Delta,N)=3$, $s=2$, $\delta=0$, and $|{\cal P}(b)|=q+1=3$. We take $\bbbf_{16}=\bbbf_2(\gamma)$ with $\gamma^4+\gamma+1=0$. Hence ${\cal P}(b)=\{1=\gamma^0,\gamma^3,1+\gamma^3=\gamma^{14}\}$. This means that $\mu_{(0)}(b)=2$, $\mu_{(1)}(b)=0$, and $\mu_{(2)}(b)=1$. Therefore, owing to Theorems \ref{teomio1} and \ref{teomio2}, $W_A=4$, $W_B=2$, $W_0^{(2)}=5$, $W_1^{(2)}=4$, $W_2^{(2)}=3$, and $\CLI$ is a $[5,4]_2$ irreducible cyclic code whose ordinary and $2$-symbol Hamming weight enumerators are $1+10T^{2}+5T^{4}$ and $1+5(T^{3}+T^{4}+T^{5})$, respectively.

\item[{\rm (c)}] Let $(q,r,N,b)=(4,3,9,2)$. Thus $u=\gcd(\Delta,N)=3$, $s=3$, $\delta=0$, and $|{\cal P}(b)|=q+1=5$. Let $\bbbf_{4}=\bbbf_2(\alpha)$ with $\alpha^2+\alpha+1=0$. We take $\bbbf_{64}=\bbbf_4(\gamma)$ with $\gamma^3+\gamma^2+\gamma+\alpha=0$. Hence ${\cal P}(b)=\{1=\gamma^0,\gamma^9,1+\gamma^9=\gamma^{27},1+\alpha\gamma^9=\gamma^{5},1+\alpha^2\gamma^9=\gamma^{40}\}$. This means that $\mu_{(0)}(b)=3$, $\mu_{(1)}(b)=1$, and $\mu_{(2)}(b)=1$. Therefore, owing to Theorems \ref{teomio1} and \ref{teomio2}, $W_A=4$, $W_B=6$, $W_0^{(2)}=6$, $W_1^{(2)}=W_2^{(2)}=7$, and $\CLI$ is a $[7,3]_4$ irreducible cyclic code whose ordinary and $2$-symbol Hamming weight enumerators are $1+21T^{4}+42T^{6}$ and $1+21T^{6}+42T^{7}$, respectively.

\item[{\rm (d)}] Let $(q,r,N,b)=(5,5,4,3)$. Thus $u=\gcd(\Delta,N)=1$ and $|{\cal P}(b)|=\mu_{(0)}(b)=q^2+q+1=31$. Therefore, owing to Theorems \ref{teomio1} and \ref{teomio2}, $W_A=625$, $W_0^{(3)}=775$, and $\CLI$ is a $[781,5]_5$ one-weight irreducible cyclic code whose ordinary and $3$-symbol Hamming weight enumerators are $1+3124T^{625}$ and $1+3124T^{775}$, respectively.
\end{enumerate}
\end{example}

\begin{remark}\label{rmone-weight} 
The previous numerical examples were corroborated with the help of a computer.
\end{remark}

As Example \ref{ejemuno}-(d) has shown, it is quite easy to obtain the $b$-symbol Hamming weight enumerator of a one-weight irreducible cyclic code (that is, when $u=1$). The following result shows it in the general case. 

\begin{theorem}\label{teomio3}
Assume the same notation as in Theorem \ref{teomio2}. If $u=\gcd(\Delta,N)=1$, then, for any $1\leq b\leq r$, the $b$-symbol Hamming weight enumerator of $\CLI$ is

$$A(T)=1+(q^r-1)T^{\frac{q^r-q^{r-b}}{N}}\;.$$
\end{theorem}

\begin{proof}
If $u=1$, then $\mu_{(0)}(b)=|{\cal P}(b)|=\frac{q^b-1}{q-1}$. Thus the result now follows from (\ref{eqwb1}). \qed
\end{proof} 

Similar to Theorem 3.3 in \cite{Zhu} we also have:

\begin{theorem}\label{teomio4}
Let $\CLI$ be as in Definition \ref{defcero}. Let $a\in \bbbf_{q^r}^*$ and consider the codeword $c(a)=(\mbox{\Tr}_{\bbbf_{q^r}/\bbbf_q}(a \gamma^{Ni}))_{i=0}^{n-1}$ in $\CLI$. Then

\begin{equation}\label{equltima}
w_r(c(a))=n\;,
\end{equation}

\noindent
and $\CLI$ is an MDS $b$-symbol code.
\end{theorem}

\begin{proof}
Suppose that $a\in {\cal C}_i^{(u,q^r)}$ for some $0 \leq i < u$. Thus, by the proof of Theorem \ref{teomio2}, $w_r(c(a))=W_i^{(r)}$ where

$$W_i^{(r)} = \frac{(q-1)q^{r/2-r}}{N}\left[|{\cal P}(r)|\left(q^{r/2}-(-1)^s\right)+(-1)^s u\mu_{((\delta-i)\!\!\!\!\!\pmod{u})}(r)\right] \ .$$

\noindent
But, owing to Lemma \ref{lemados}, $\mu_{((\delta-i)\!\!\!\!\!\pmod{u})}(r)=\frac{\Delta}{u}$. On the other hand, $|{\cal P}(r)|=\Delta=\frac{q^r-1}{q-1}$ and $n=\frac{q^r-1}{N}$. Thus, (\ref{equltima}) now follows. Finally, since $d_b(\CLI)=n$ and $|\CLI|=q^r$, $\CLI$ is an MDS $b$-symbol code. \qed
\end{proof}

\end{document}